\documentclass[12pt]{iopart}
\usepackage{iopams}
\usepackage{setstack}
\usepackage{graphicx}
\usepackage{amsthm}
\newtheorem{theorem}{Theorem}

\newtheorem{lemma}{Lemma}

\theoremstyle{remark}
\newtheorem*{remark}{Remark}

\begin{document}
\newcommand{\real}{\textrm{Re}\:}
\newcommand{\sto}{\stackrel{s}{\to}}
\newcommand{\supp}{\textrm{supp}\:}
\newcommand{\wto}{\stackrel{w}{\to}}
\newcommand{\ssto}{\stackrel{s}{\to}}
\newcounter{foo}
\providecommand{\norm}[1]{\lVert#1\rVert}
\providecommand{\abs}[1]{\lvert#1\rvert}

\title{Zero Energy Bound States and Resonances in Three--Particle Systems. }

\author{Dmitry K. Gridnev\footnote[1]{On leave from: Institute of Physics, St. Petersburg State
University, Ulyanovskaya 1, 198504 Russia}}

\address{FIAS, Ruth-Moufang-Stra{\ss}e 1, D--60438 Frankfurt am Main, Germany}
\ead{gridnev@fias.uni-frankfurt.de}

\begin{abstract}
We consider a three--particle system in $\mathbb{R}^3$ with non--positive
pair--potentials and non--negative essential spectrum. Under certain
restrictions on potentials it is proved that the eigenvalues are absorbed at
zero energy threshold given that there is
no negative energy bound states and zero energy resonances in particle pairs. It
is shown that the condition on the absence of zero energy resonances in particle
pairs is essential. Namely, we prove that if at least one pair of particles has
a zero energy resonance then a square integrable zero energy ground state of
three particles does not exist. It is also proved that one can tune the coupling
constants of pair potentials so that for any given $R, \epsilon >0$: (a) the
bottom of the essential spectrum is at zero; (b) there is a negative energy
ground state $\psi(\xi)$ such that $\int |\psi(\xi)|^2 d^6 \xi= 1$ and
$\int_{|\xi| \leq R} |\psi(\xi)|^2 d^6 \xi < \epsilon$.
\end{abstract}


\maketitle


\section{Introduction}\label{sec:1}

Let us consider an $N$--particle Schr\"odinger operator
\begin{equation}\label{e111}
    H(\lambda) = H_0 - \lambda \sum_{1 \leq i<j \leq N} V_{ij}(r_i - r_j),
\end{equation}
where $\lambda > 0$ is a coupling constant, $H_0$ is a kinetic energy operator
with the centre of mass removed, $r_i \in \mathbb{R}^3$ are particle position
vectors, the pair potentials are real (further restrictions on the potentials
would be given later). Suppose that for $\lambda$  in the vicinity of some
$\lambda_{cr} < \infty$ there is a bound state $\psi(\lambda) \in D(H_0)$ with
the energy $E(\lambda) < \inf \sigma_{ess} (H(\lambda))$ and $E(\lambda) \to
\inf \sigma_{ess}(H(\lambda_{cr}))$ when $\lambda \to \lambda_{cr}$. The
question whether $E(\lambda_{cr})$ is an eigenvalue of $H(\lambda_{cr})$ was
considered in various contexts in \cite{klaus1,klaus2,simon,ostenhof,gest,prl}
(the list of references is by far incomplete).

In \cite{karner}, Theorem~3.3,  it was claimed that if $V_{ij} \in C_0^\infty
(\mathbb{R}^3)$, $V_{ij} \geq 0$, and none of the subsystems has negative energy
bound states or zero energy resonances, then there exists $\psi(\lambda_{cr})
\in D(H) = D(H_0), \psi(\lambda_{cr}) \neq 0$ such that $H(\lambda_{cr})
\psi(\lambda_{cr})= 0$. Unfortunately, the proof in \cite{karner} contains a
mistake. In Eq.~53 of \cite{karner} the mixed term containing first order
derivatives is erroneously omitted, which makes the results of Ref.~35 in
\cite{karner} concerning the fall off of the wave function inapplicable. And it
is not immediately clear how this mistake can be corrected. Here we prove
the result stated by Karner for $N=3$ with a different method and for a larger
class of potentials (Theorem~\ref{th:2} of this paper).

Naturally, a question can be raised whether the condition on the absence of
2--particle zero energy resonances is essential. Here we show that indeed it is.
Namely, in Sec.~\ref{Q@sec:3} (Theorem~\ref{Q@th:1}) we prove that the
3--particle ground state at zero energy can be at most a resonance and not a
$L^2$ state if at least one pair of particles has a resonance at zero energy. 
The method of proof is inspired by \cite{klaus1,sobol,yafaev}. 
The last section provides an example of a 3--particle system, where each
2--particle subsystem is unbound, one 2--particle subsystem is at the
2--particle coupling constant threshold and the whole 3--particle system has a
resonance but not a bound state at zero energy. Systems like this can always be
constructed through appropriate tuning of the coupling constants.

Note, that the 3--particle case differs essentially from the 2--particle case,
where
under similar restrictions on pair potentials the zero energy ground
state cannot be a bound state \cite{klaus1,prl}. This difference
has far reaching physical consequences, which concern the size of
a system in its ground state (we ignore the particle statistics here). In the
two--particle case the size of a system in the ground state can be made infinite
by tuning, for example, the coupling constant so that the bound state with
negative energy approaches the zero energy threshold \cite{prl}. In the three
particle case the size of the system remains finite, given that in the course of
tuning the coupling constants of two--particle subsystems stay away from
critical
values, at which the two--particle zero energy resonances appear. To underline
the
connection with the size of the system we formulate the proofs in terms of
spreading and non--spreading sequences of bound states. The obtained results are relevant 
in the physics of halo nuclei \cite{vaagen}, molecular physics
\cite{fedorov} and Efimov states \cite{efimov}. Here we refer the reader to the last section, where we discuss 
possible physical applications of our results.

The paper is organized as follows. In Sec.~\ref{sec:2} we use the ideas of
Zhislin \cite{zhislin} to set up the framework for the analysis of eigenvalue
absorption in connection with the spreading of sequences of wave functions. Here
we prefer to maintain generality and do not restrict ourselves to $N=3$. In
Sec.~\ref{sec:3} we consider the 3--particle case and employ the equations of
Faddeev type to prove Theorem~\ref{th:2}, which is the main result of this section. In
Sec.~\ref{Q@sec:2} we prove an auxiliary result 
concerning the two--particle zero energy resonance. In Sec.~\ref{Q@sec:3} we
prove Theorem~\ref{Q@th:1}, which says that the ground state of three particles 
cannot be a bound state given that there is no bound states in particle pairs
and at least one particle pair has a zero energy resonance. The last section 
provides a constructive example of a critically bound three--particle system
with such conditions. In the last section we also give the overview of relevant physical phenomena and discuss possible applications of our results.

\section{Spreading and Bound States at Threshold}\label{sec:2}

The main result of this section (Theorem~\ref{th:1}) appears implicitly in
\cite{zhislin}, where Zhislin considers minimizing sequences of the energy
functional in Sobolev spaces. For our purposes it is more useful to consider
sequences of eigenstates and use an approach in the spirit of \cite{simon}.

Consider the $N$-particle Hamiltonian, which depends on a parameter
\begin{eqnarray}
    H(\lambda) = H_0 + V(\lambda) ,  \label{xc31}  \\
V(\lambda) =  \sum_{1 \leq i<j \leq N} V_{ij} (\lambda; r_i - r_j )
\label{:xc31},
\end{eqnarray}
where $H_0$ is the kinetic energy operator with the centre of mass removed, $V_{ij}$ are pair potentials and 
$r_i \in \mathbb{R}^3$ are position vectors.
For the parameter $\lambda$ we assume that $\lambda \in \mathbb{R}$ (this is
done for clarity, in fact, $\lambda$ can take values in a topological space). We
impose the following set of restrictions.
\begin{list}{R\arabic{foo}}
{\usecounter{foo}
    \setlength{\rightmargin}{\leftmargin}}
\item $H(\lambda)$ is defined for an infinite sequence of parameter values
$\lambda_1, \lambda_2, \ldots $ and $\lambda_{cr}$, where $\lim_{n\to \infty}
\lambda_n = \lambda_{cr}$.

\item $|V_{ij} (\lambda; y )| \leq F (y)$ for all $\lambda$ defined in R1, where
$V_{ij}, F \in L^2 (\mathbb{R}^3) + L_\infty^\infty (\mathbb{R}^3)$.

\item $\forall f \in C^\infty_0 (\mathbb{R}^{3N-3} )\colon  \lim_{\lambda_n \to
\lambda_{cr}} \bigl\| \bigl[ V(\lambda_n) - V(\lambda_{cr}) \bigr] f \bigr\| =
0$.
\end{list}
The symbol $L_\infty^\infty$ denotes bounded Borel functions going to zero at
infinity. By R2 $H(\lambda)$ is self-adjoint on $D(H_0 )$ \cite{reed}.

The bottom of the essential spectrum of $H(\lambda)$ is denoted as
\begin{equation}\label{sigmaess}
  E_{thr} (\lambda) := \inf \sigma_{ess} (H(\lambda)) .
\end{equation}

The set of requirements on the system continues as follows
\begin{list}{R\arabic{foo}}
{\usecounter{foo}
    \setlength{\rightmargin}{\leftmargin}}
\setcounter{foo}{3}
\item
for all $ \lambda_n $ there is $E(\lambda_n) \in \mathbb{R}, \psi(\lambda_n) \in
D(H_0)$ such that $H(\lambda_n) \psi(\lambda_n) = E(\lambda_n) \psi(\lambda_n)$,
where $\| \psi(\lambda_n) \| = 1$ and $E(\lambda_n) < E_{thr}(\lambda_n)$.

\item
$\lim_{\lambda_n \to \lambda_{cr}} E(\lambda_n) = \lim_{\lambda_n \to
\lambda_{cr}} E_{thr}(\lambda_n) = E_{thr} (\lambda_{cr})$ .
\end{list}

The requirements R4-5 say that for each $n$ the system has a level below the
continuum and for $\lambda_n \to \lambda_{cr}$ the energy of this level
approaches the bottom of the continuous
spectrum.

In the proofs we shall use the term ``spreading sequence'', which is due to
Zhislin \cite{zhislin}. The sequence of functions $f_n (x) \in L^2
(\mathbb{R}^d)$
\textbf{spreads} if there is $a>0$ such that $\limsup_{n \to \infty} \|
\chi_{\{x||x| > R\}} f_n \| > a$ for all $R>0$. (the notation $\chi_{\Omega}$
always
means the characteristic function of the set $\Omega$). The sequence $f_n$ is
called
\textbf{totally spreading}  if $\lim_{n \to \infty} \|  \chi_{\{x||x| \leq R\}}
f_n \| = 0$ for all $R >0$. Note that any normalized sequence,
which converges in norm, does not spread, and any sequence, which goes to zero
in norm, totally spreads.

\begin{lemma}\label{lem:1}
Let $H(\lambda)$ be a Hamiltonian satisfying R1-5. Then
\begin{equation}
\sup_n \| H_0 \psi(\lambda_n)\| < \infty .
\end{equation}
\end{lemma}

\begin{proof}
The statement represents a
well-known fact, see e. g.  \cite{zhislin} but for
completeness we give the proof right here.
The Schr\"odinger equation $H_0 \psi (\lambda_n)= -V(\lambda_n)\psi (\lambda_n)
+ E(\lambda_n)
\psi(\lambda_n)$ gives the bound $\| H_0 \psi (\lambda_n) \| \leq \| V(\lambda_n
)\psi (\lambda_n) \| + \mathcal{O}(1)$. It remains to show that $\| V(\lambda_n
)\psi (\lambda_n) \| = \mathcal{O}(1)$.
By R2 $|V_{ij}| \leq F_{ij}$, where for a shorter notation we denote $F_{ij} :=
F(r_i - r_j)$. Using
that $F_{ij}$ is $H_0$--bounded
with a relative bound 0 (see vol.2, Theorem X.16 in \cite{reed}) we obtain
\begin{eqnarray}
    \|V(\lambda_n) \psi (\lambda_n)\| = \bigl\| \sum_{i<j} V_{ij} (\lambda_n;
r_i - r_j )  \psi (\lambda_n) \bigr\|
    \leq \sum_{i<j} \bigl\| F_{ij}  \psi (\lambda_n) \bigr\|\nonumber\\
 \leq a \|H_0 \psi (\lambda_n)\| + b \leq
    a \| V(\lambda_n)\psi (\lambda_n)\| + \mathcal{O}(1) , \label{kuka:1}
\end{eqnarray}
where $a, b >0$ are constants and $a$ can be
chosen as small as pleased. Setting, for example, $a = 1/2$ gives $\|
V(\lambda_n
)\psi (\lambda_n) \| = \mathcal{O}(1)$.  \end{proof}

The following theorem illustrates the connection between non-spreading and bound
states at threshold.

\begin{theorem}[Zhislin]\label{th:1}
Let $H(\lambda)$ satisfy R1-5. If the sequence  $\psi (\lambda_n)$ does not
totally spread then $H(\lambda_{cr})$ has a  bound state at threshold $\psi_{cr}
\in D(H_0)$, so that
\begin{equation}\label{xc11}
H (\lambda_{cr}) \psi_{cr} = E_{thr}(\lambda_{cr}) \psi_{cr} ,
\end{equation}.
\end{theorem}
For the proof we need a couple of technical Lemmas.

\begin{lemma}\label{lem:2}
Suppose $f_n \in D(H_0)$ is such that $\sup_n \| H_0 f_n\| < \infty$ and $f_n
\wto f_0$. Then (a) $f_0 \in D(H_0)$; (b) for any operator $A$, which is
relatively $H_0$--compact $\| A ( f_n - f_0 ) \| \to 0$.
\end{lemma}
\begin{proof}
First, let us prove that the sequence $H_0 f_n $ is weakly convergent. A proof
is by contradiction. Suppose $H_0 f_n $ has two weak limit points, {\it i.e.}
there exist $f'_k, f''_k$, which are subsequences of $f_n$ and for which $H_0
f'_k \wto \phi_1$ and $H_0 f''_k \wto \phi_2$, where $\phi_{1,2} \in L^2$ and
$\phi_1 \neq \phi_2$. On one hand, because $\phi_1 \neq \phi_2$ and $D(H_0)$ is
dense in $L^2$ there is $g \in D(H_0)$ such that $(\phi_1 - \phi_2 , g) \neq 0$.
On the other hand, using that $f'_k \wto f_0$ and  $f''_k \wto f_0$ we get
\begin{equation}\label{:labi:}
    (\phi_1 - \phi_2 , g) = \lim_{k \to \infty} \left[ \bigl( H_0 (f'_k -
f''_k), g\bigr)\right] =
    \lim_{k \to \infty} \left[ \bigl( (f'_k - f''_k), H_0 g\bigr)\right] = 0,
\end{equation}
a contradiction. Hence, $H_0 f_n \wto G$, where $G \in L^2$. $\forall f \in
D(H_0)$ by self-adjointness of $H_0$ we obtain $(H_0 f, f_0) = \lim_{n \to
\infty} (H_0 f, f_n) = (f, G)$. Thus $f_0 \in D(H_0)$ and $G = H_0 f_0$, which
proves (a). To prove (b) note that $(H_0 + 1)(f_n - f_0 ) \wto 0$. Using that
compact operators acting on weakly convergent sequences make them converge in
norm we get
\begin{equation}\label{fck}
A(f_n - f_0 )  = A(H_0 + 1)^{-1} (H_0 + 1) (f_n - f_0 )\to 0 ,
\end{equation}
since $A(H_0 + 1)^{-1}$ is compact by condition of the lemma.   \end{proof}

\begin{lemma}\label{lem:3}
Suppose $f_n \in D(H_0)$ is such that $\sup_n \| H_0 f_n\| < \infty$ and $f_n
\wto f_0$.
Then (a) if $f_n$ does not spread then $f_n \to f_0$ in norm; (b) if $f_n$ does
not totally spread then $f_0 \neq 0$.
\end{lemma}
\begin{proof}
Let us start with (a).
Because $f_n$ does not spread it is enough to show that $\| \chi_{\{x| |x| \leq
R\}} (f_n - f_0) \| \to 0$ for all $R$ in norm. And this is true because
$\chi_{\{x| |x| \leq R\}}$ is relatively $H_0$--compact \cite{reed,teschl} and
Lemma~\ref{lem:2} applies. To prove (b) let us assume by contradiction that $f_n
\wto 0$. Using the same arguments we get that $\| \chi_{\{x| |x| \leq R\}} f_n
\| \to 0$ for all $R$. But this would mean that $f_n$ totally spreads contrary
to the condition of the Lemma.   \end{proof}

\begin{proof}[Proof of Theorem~\ref{th:1}]
Without loosing generality we can assume that there are $a, R > 0$ such that $\|
\chi_{\{x| |x| < R\}} \psi(\lambda_n) \| > a$
(otherwise we can pass to an appropriate subsequence, since $\psi (\lambda_n)$
does not totally spread). By the
Banach-Alaoglu theorem we choose a weakly convergent subsequence so that
$\psi(\lambda_{n_k}) \wto
\psi_{cr}$, where $\psi_{cr} \in D(H_0)$ by Lemma~\ref{lem:2}.
$\psi(\lambda_{n_k})$ does not totally spread and is weakly
convergent, hence,  by Lemma~\ref{lem:3}(b) $\psi_{cr} \neq 0$. For any $f \in
C_{0}^\infty$ we have
\begin{eqnarray}
\fl \Bigl([H(\lambda_{cr}) - E_{thr}(\lambda_{cr})]f , \psi_{cr} \Bigr) =
\lim_{\lambda_{n_k} \to \lambda_{cr}} \Bigl( [H(\lambda_{cr})
-E_{thr}(\lambda_{n_k})]f , \psi (\lambda_{n_k}) \Bigr) \nonumber\\
\fl = \lim_{\lambda_{n_k} \to \lambda_{cr}} \Bigl( \bigl[ H(\lambda_{n_k}) -
(V(\lambda_{n_k})-V(\lambda_{cr}) )  -E_{thr}(\lambda_{n_k}) \bigr] f ,\psi
(\lambda_{n_k}) \Bigr)  \nonumber \\
\fl = \lim_{\lambda_{n_k} \to \lambda_{cr}} \Bigl\{ \bigl[E(\lambda_{n_k})
-E_{thr}(\lambda_{n_k}) \bigr] \Bigl( f ,\psi (\lambda_{n_k}) \Bigr)  -
\Bigl([V(\lambda_{n_k})-V(\lambda_{cr}) ]f ,\psi (\lambda_{n_k})\Bigr) \Bigr\} =
0, 
\end{eqnarray}
where in the last equation we have used R3, R5.
Summarizing, for all $f \in C_{0}^\infty$ we have
\begin{equation}\label{xc12}
    \left(\bigl[H(\lambda_{cr}) - E_{thr}(\lambda_{cr})\bigr]f , \psi_{cr}
\right) =  \left(f , \bigl[H(\lambda_{cr}) -
E_{thr}(\lambda_{cr})\bigr]\psi_{cr} \right) = 0,
\end{equation}
meaning that (\ref{xc11}) holds.  \end{proof}

The following Lemmas will be needed in the next Section.
\begin{lemma}\label{lem:4}
 A uniformly norm--bounded sequence of functions $f_n \in L^2 (\mathbb{R}^n)$
having a property that
every weakly convergent subsequence converges also in norm  does not
spread.
\end{lemma}
\begin{proof}
By contradiction, let us assume that $f_n$ spreads. Then it is
possible to extract a subsequence $g_k = f_{n_k}$ with the
property $\| \chi_{\{x| |x| \geq k \}} g_k \| > a$, where $a > 0$
is a constant. On one hand, it is easy to see that $g_k $ has no subsequences
that converge in
norm. On the other hand, by the Banach-Alaoglu theorem $g_k$ must have at least
one  weakly converging subsequence, which is norm--convergent by condition of
the lemma. A contradiction. \end{proof}

\begin{lemma}\label{lem:5}
Suppose $g \in C (\mathbb{R}^{3N-3})$ has the property that $|g|\leq 1$ and
$g=0$ if $|r_i - r_j | < \delta |x|$, where $\delta$ is a constant. Then the
operator $gF(r_i - r_j)$ is relatively $H_0$--compact.
\end{lemma}
\begin{proof}
It suffices to consider the case $F\in L^2 (\mathbb{R}^3)$ (the case $F \in
L_\infty^\infty (\mathbb{R}^3)$ trivially follows from Lemma~7.11 in
\cite{teschl}). For $k=1,2,\ldots$ we can write
\begin{equation}\label{muta}
\fl gF_{ij}(H_0 + 1)^{-1} = \chi_{\{x|\: |r_i - r_j | <k\}}gF_{ij}(H_0 + 1)^{-1} +
\chi_{\{x|\: |r_i - r_j | \geq k\}}gF_{ij}(H_0 + 1)^{-1} ,
\end{equation}
where again $F_{ij}:= F(r_i - r_j)$.
The first operator on the rhs is compact (Lemma~7.11 in \cite{teschl}). We need
to show that the second one goes to zero in norm when $k \to \infty$ (in this
case the operator on the lhs is compact as a norm-limit of compact operators).
The following integral estimate of the square of its norm is trivial
\begin{equation}\label{muta2}
\fl \bigl\| \chi_{\{x|\: |r_i - r_j | \geq k\}}gF_{ij}(H_0 + 1)^{-1}\bigr\|^2 \leq
\frac 1{(4\pi)^2}
\int_{|r|\geq k} d^3 r\;  |F(r)|^2 \int d^3 r' \; \frac{e^{-2|r'|}}{|r'|^2} .
\end{equation}
Because $F\in L^2 (\mathbb{R}^3)$ the rhs goes to zero as $k \to \infty$.
\end{proof}

\section{Zero Energy Bound States of Three Particles}\label{sec:3}

We apply the framework of Sec.~\ref{sec:2} to the system of three particles with
non-positive potentials. The case $N>3$ and potentials taking both signs would
be considered elsewhere. For simplicity we take the parameter $\lambda > 0$ as a
coupling constant of the interaction (see \cite{klaus1,klaus2})
\begin{eqnarray}
    H(\lambda) = H_0 - \lambda V  , \label{xc31aa} \\
V =  \sum_{1 \leq i<j \leq 3} V_{ij} (r_i - r_j ) . \label{:xc31aa}
\end{eqnarray}
We shall need the following additional requirements
\begin{list}{R\arabic{foo}}
{\usecounter{foo}
    \setlength{\rightmargin}{\leftmargin}}
\setcounter{foo}{5}
\item
$V_{ij} \geq 0$ and $\lambda V_{ij} (y) \leq F(y)$, where $F \in L^2
(\mathbb{R}^3) \cap L^1 (\mathbb{R}^3) $ and $\lambda$ takes values as defined
in R1.
\item
There exists $\epsilon > 0$ such that $H_0 - (\lambda + \epsilon )V_{ij} \geq 0$
for all $\lambda$ defined in R1 and all pair potentials $V_{ij}$.
\end{list}
Requirement R7 means that the two--particle subsystems have no bound states with
negative energy and no resonances at zero energy. This results in
$E_{thr}(\lambda) = 0$. Our aim is to prove
\begin{theorem}\label{th:2}
Suppose $H(\lambda)$ defined in (\ref{xc31aa})--(\ref{:xc31aa}) satisfies R1,
R4-7. Then for $n \to \infty$ the sequence $\psi(\lambda_n)$
does not spread and there exists a bound state at threshold $\psi_{cr} \in
D(H_0)$, $\|\psi_{cr}\| \neq 0$, such that $H(\lambda_{cr})\psi_{cr} = 0$.
\end{theorem}
We shall defer the proof, which boils down to the construction of Faddeev
equations \cite{faddeev}, see also \cite{sobol,yafaev}, to the end of the
section.
Let us introduce an analytic operator function $B_{ij} (z)$ for each pair of
particles $(ij)$. We shall construct $B_{12}$ and the other two operators are
constructed similarly. We use Jacobi coordinates \cite{greiner} $x = [\sqrt{2
\mu_{12}}/\hbar](r_2 - r_1)$ and
$y = [\sqrt{2 M_{12}}/\hbar](r_3 - m_1/(m_1+m_2) r_1 - m_2/(m_1+m_2) r_2)$,
where $\mu_{ij} = m_i m_j /(m_i + m_j)$ and $M_{ij} = (m_i + m_j)m_l / (m_i +
m_j + m_l)$
are reduced masses (the indices $i,j,l$ are all different). These coordinates
make the kinetic energy operator take the form
\begin{equation}\label{ay4}
    H_0 = - \Delta_x - \Delta_y .
\end{equation}
Let $\mathcal{F}_{12}$ denote the partial Fourier transform in
$L^2(\mathbb{R}^6)$ acting as follows
\begin{equation}\label{ay5}
\hat f(x,p_y)  =  \mathcal{F}_{12} f(x,y) = \frac 1{(2 \pi )^{3/2}} \int d^3 y
\; e^{-ip_y \cdot \; y} f(x,y) .
\end{equation}
Then $B_{12}(z)$ is defined through
\begin{equation}\label{ay6}
B_{12}(z) = 1 + z  + \mathcal{F}^{-1}_{12} t(p_y) \mathcal{F}_{12},
\end{equation}
where
\begin{equation}\label{ay669}
t(p_y) = (\sqrt{|p_y|} - 1)\chi_{\{p_y | \; |p_y| \leq 1\}}.
\end{equation}
Similarly, using other Jacobi coordinates one defines $B_{ij}(z)$  and
$\mathcal{F}_{ij}(z)$ for all particle pairs.
Note that $B_{ij}(z)$ and $B^{-1}_{ij}(z)$ are analytic on $\real z > 0$.

\begin{lemma}\label{lem:6}
The operator function in $L^2 (\mathbb{R}^6)$
\begin{equation}\label{ya23}
  \mathcal{A}_{ij} (z) = (H_0 + z^2)^{-1} V_{ij}^{1/2} B_{ij} (z)
\end{equation}
is uniformly bounded for $z \in (0,1]$, and strongly continuous for $z \to +0$.
\end{lemma}
\begin{proof}
We consider the case when $(ij) = (12)$, other indices are treated similarly.
Instead of $\mathcal{A}_{12} (z) $ we consider  $\mathcal{F}_{12}
\mathcal{A}_{12} (z) \mathcal{F}^{-1}_{12}   $. We take $z \in (0,1)$ and split
the operator
\begin{equation}\label{redew1}
    \mathcal{F}_{12}  \mathcal{A}_{12} (z) \mathcal{F}^{-1}_{12}   = K_1 (z) +
K_2 (z) ,
\end{equation}
where
\begin{eqnarray}
K_1 (z) = (- \Delta_x + p_y^2 + z^2)^{-1} V^{1/2}_{12}(\alpha x) [t(p_y) + 1]  ,
\label{redew2} \\
K_2 (z) = (- \Delta_x + p_y^2 + z^2)^{-1} V^{1/2}_{12}(\alpha x) z
\label{redew3}
\end{eqnarray}
are integral operators acting on $\phi (x, p_y)\in L^2(\mathbb{R}^6)$ as
\begin{eqnarray}
K_1 (z)\phi  = \frac 1{4 \pi} \int d^3 x'  \frac{e^{-\sqrt{p_y^2 + z^2}
|x-x'|}}{|x-x'|}  V^{1/2}_{12}(\alpha x') [t(p_y) + 1] \phi (x', p_y) ,
\label{redew6} \\
K_2 (z)\phi  = \frac{z }{4 \pi} \int d^3 x'  \frac{e^{-\sqrt{p_y^2 + z^2}
|x-x'|}}{|x-x'|}  V^{1/2}_{12}(\alpha x') \phi (x', p_y) . \label{redew7}
\end{eqnarray}
The numerical coefficient $\alpha$ depends on masses $\alpha := \hbar
/\sqrt{2\mu_{12}} $. Applying the Cauchy-Schwarz inequality we get
\begin{eqnarray}
\fl \bigl| K_1 (z) \phi \bigr|^2 \leq \int d^3 x'  \frac{e^{-2 |p_y|
|x-x'|}}{|x-x'|^2} [t(p_y) +1]^2 V_{12} (\alpha x')  \times \int d^3 x' \bigl|
\phi (x', p_y) \bigr|^2  ,  \label{redew4} \\
\fl \bigl| K_2 (z)\phi \bigr|^2 \leq z^2  \int d^3 x'  \frac{e^{-2 z
|x-x'|}}{|x-x'|^2} V_{12} (\alpha x')  \times \int d^3 x' \bigl| \phi (x', p_y)
\bigr|^2 \label{redew5} ,
\end{eqnarray}
where we have used $ z \in (0,1]$. Integrating (\ref{redew4}) and (\ref{redew5})
over $x$ leads to
\begin{eqnarray}
\int d^3 x \bigl| K_1 (z) \phi \bigr|^2 \leq cc'c''\Bigl[\int d^3 x' \bigl| \phi
(x', p_y) \bigr|^2 \Bigr] , \label{wadi1}\\
\int d^3 x \bigl| K_2 (z) \phi \bigr|^2 \leq cc'\Bigl[\int d^3 x' \bigl| \phi
(x', p_y) \bigr|^2 \Bigr] , \label{wadi2}
\end{eqnarray}
where $c, c', c''$ are the following finite constants
\begin{eqnarray}
    c = \int d^3 x'  V_{12} (\alpha x') , \label{aijacts9} \\
    c' = \int d^3 x \frac{e^{-2 |x|}}{|x|^2}  , \label{aijacts10} \\
    c'' = \sup_{p_y \in \mathbb{R}^3} [t(p_y) + 1 ]^2 / |p_y| . \label{aijacts4}
\end{eqnarray}
Integrating (\ref{wadi1})--(\ref{wadi2}) over $p_y$ gives that $K_{1,2} (z)$ is
uniformly norm--bounded for $z \in (0,1]$. The strong continuity for $z \to +0$
follows from (\ref{redew6})--(\ref{redew7}) by the dominated convergence
theorem. \end{proof}

For $z \in (0,\infty)$ let us define
\begin{equation}\label{ay2}
    \mathcal{C}_{ik;jm}(z) = V^{1/2}_{ik} (H_0 + z^2)^{-1} V^{1/2}_{jm} .
\end{equation}
The properties of the above operator are summarized in the following 
\begin{lemma}\label{lem:7,5}
Suppose $H(\lambda)$ defined in (\ref{xc31aa})--(\ref{:xc31aa}) satisfies R1,
R6, R7. Then (a) the operator function
$\mathcal{C}_{ik;jm}(z)$ is norm--continuous for $z >0$ and has a norm limit for
$z \to +0$; (b) there exists $\delta >0$ such that
$\lambda_n \| C_{ik;ik} (z) \| < 1-\delta$ for all $z \geq 0$.
\end{lemma}
\begin{proof}
Below we prove that $\mathcal{C}_{ik;jm}(z)$ for $z_{1,2} \in (0,\infty)$
satisfies the following continuity condition
\begin{equation}\label{sukref7}
 \| \mathcal{C}_{ik;jm}(z_1) - \mathcal{C}_{ik;jm}(z_2) \| \leq
l\sqrt{\bigl|z^2_1 - z^2_2 \bigr|},
\end{equation}
where $l$ is a constant. From (\ref{sukref7}) it easily follows that
$\mathcal{C}_{ik;jm}(z)$ for $z \to +0$ form a Cauchy sequence. Therefore we can
define the norm limit
\begin{equation}
 \mathcal{C}_{ik;jm}(0) := \lim_{z \to +0}   \mathcal{C}_{ik;jm}(z)
\end{equation}
and $\mathcal{C}_{ik;jm}(z)$ becomes norm--continuous for $z \geq 0$. Let us
first
prove (\ref{sukref7}) for $\mathcal{C}_{ik;ik}(z)$. It suffices to consider
$\mathcal{C}_{12;12}(z)$. Taking $0 < z_1 < z_2$ we have
\begin{equation}\label{sukref8}
 \bigl\| \mathcal{C}_{12;12}(z_1) - \mathcal{C}_{12;12}(z_2) \bigr\| = \| K(z_1,
z_2)\|,
\end{equation}
 where
\begin{equation}
 K(z_1, z_2) := \mathcal{F}_{12} V^{1/2}_{12} \Bigl[ (H_0 + z^2_1)^{-1} - (H_0 +
z^2_2)^{-1}\Bigr] V^{1/2}_{12} \mathcal{F}_{12}^{-1} .
\end{equation}
The integral operator $ K(z_1, z_2)$ acts on $\phi(x, p_y) \in L^2
(\mathbb{R}^6)$ as
\begin{equation}
 K(z_1, z_2) \phi (x, p_y) = \int d^3 x' \: k(z_1, z_2; x, x', p_y) \phi(x',
p_y),
\end{equation}
where
\begin{equation}
\fl k(z_1, z_2; x, x', p_y) = \frac{\sqrt{V_{12}(\alpha x) V_{12}(\alpha x')}}{4\pi
|x-x'|} \Bigl\{ e^{-\sqrt{p_y^2 + z_1^2} |x-x'|} - e^{-\sqrt{p_y^2 + z_2^2}
|x-x'|} \Bigr\}.
\end{equation}
Obviously,
\begin{equation}\label{sukref2}
 \| K(z_1, z_2) \|^2 \leq \sup_{p_y \in \mathbb{R}^3} \int d^3 x \: d^3 x' \:
\bigl|  k(z_1, z_2; x, x', p_y)  \bigr|^2 .
\end{equation}
Using the inequality
\begin{equation}
|k(z_1, z_2; x, x', p_y)| \leq  \frac{\sqrt{V_{12}(\alpha x) V_{12}(\alpha
x')}}{4\pi } \sqrt{z_2^2 - z_1^2}
\end{equation}
we obtain from (\ref{sukref2})
\begin{equation}\label{sukref4}
 \| K(z_1, z_2) \|^2 \leq \frac{c^2}{16\pi^2}\bigl|z_2^2 - z_1^2\bigr| ,
\end{equation}
where $c$ was defined in (\ref{aijacts9}). From (\ref{sukref4}) and
(\ref{sukref8}) the continuity condition (\ref{sukref7}) follows for
$\mathcal{C}_{ik;ik}(z)$.
It remains to prove (\ref{sukref7}) for $\mathcal{C}_{ik;jm}(z)$.
For $0 < z_1 < z_2$ by the resolvent identity we have
\begin{equation}\label{sukref1}
 \fl (H_0 + z^2_1)^{-1} - (H_0 + z^2_2)^{-1} = (z_2^2 - z_1^2)(H_0+z_1^2)^{-1/2}
(H_0+z_2^2)^{-1} (H_0+z_1^2)^{-1/2} \geq 0
\end{equation}
Thus we can write
\begin{eqnarray}
\bigl\| \mathcal{C}_{ik;jm}(z_1) - \mathcal{C}_{ik;jm}(z_2) \bigr\| \nonumber \\
\fl  = \Bigl\| V^{1/2}_{ik} \bigl[(H_0 + z^2_1)^{-1} - (H_0 + z^2_2)^{-1}\bigr]^{1/2}
\bigl[(H_0 + z^2_1)^{-1} - (H_0 + z^2_2)^{-1}\bigr]^{1/2} V^{1/2}_{jm}
\Bigr\| \nonumber\\
\fl \leq \bigl\| \mathcal{C}_{ik;ik}(z_1) - \mathcal{C}_{ik;ik}(z_2) \bigr\|^{1/2}
\bigl\| \mathcal{C}_{jm;jm}(z_1) - \mathcal{C}_{jm;jm}(z_2) \bigr\|^{1/2} \leq
l\sqrt{\bigl|z^2_1 - z^2_2 \bigr|} ,
\end{eqnarray}
where we have used that $\|AB \| \leq \|AA^\dagger \|^{1/2} \|BB^\dagger
\|^{1/2}$ for any bounded $A,B$.

Let us prove (b). The statement follows from the Birman--Schwinger principle,
see \cite{klaus1}. For completeness we sketch the proof here. By R7 for all $z
>0$ we have
\begin{equation}
 H_0 + z^2 \geq (\epsilon + \lambda_n)V_{ik}
\end{equation}
Forming a scalar product with $(H_0 + z^2)^{-1/2}\eta$, where $\eta \in D(H_0)$,
$\|\eta\| = 1$ gives
\begin{equation}
(\epsilon + \lambda_n)\bigl(\eta, (H_0+z^2)^{-1/2}V_{ik} (H_0+z^2)^{-1/2} \eta
\bigr) \leq 1
\end{equation}
This means that
\begin{equation}
\bigl\| \mathcal{C}_{ik;ik}(z)\bigr\|  =  \bigl\|(H_0+z^2)^{-1/2}V_{ik}
(H_0+z^2)^{-1/2} \bigr\|\leq 1/(\epsilon + \lambda_n),
\end{equation}
where we have used that $\|AA^\dagger\| = \|A^\dagger A\| $ for any bounded $A$.
Thus (b) follows if we set $\delta = \epsilon/(\overline{\lambda} + \epsilon)$,
where $\overline{\lambda} := \sup_n \lambda_n$.
\end{proof}

We shall need the following
\begin{lemma}\label{lem:7}
Suppose R1, R4-7 are satisfied. Then the operators
\begin{equation}\label{ay2222}
    \mathcal{R}_{ij}(\lambda_n) = [1 - \lambda_n \mathcal{C}_{ij;ij}(k_n)]^{-1}
\quad \mathrm{for } \quad k_n = \sqrt{|E(\lambda_n)|}
\end{equation}
are uniformly bounded for all $n$ and converge in norm when $n \to \infty$.
\end{lemma}
\begin{proof}
By the previous Lemma the operators $\mathcal{C}_{ij;ij}(k_n)$ converge in norm
to  $\mathcal{C}_{ij;ij}(0) $ and $\lambda_n \| C_{ij;ij} (k_n) \| < 1-
\delta$,
where $\delta > 0$ is a constant. Now the result follows from expanding
(\ref{ay2222}) in von Neumann series. \end{proof}

\begin{lemma}\label{lem:8}
For $(ik) \neq (jm)$ the operator function $B^{-1}_{ik} (z) \mathcal{C}_{ik;jm}
(z) $
is uniformly norm--bounded for $z \in (0,1]$ and strongly continuous for $z \to
+0$.
\end{lemma}
\begin{proof}
We focus on $B^{-1}_{12} (z) \mathcal{C}_{12;23} (z) $, the other indices are
treated similarly. Let us show that $\mathcal{F}_{12} B^{-1}_{12} (z)
\mathcal{C}_{12;23} (z)\mathcal{F}_{12}^{-1}$ is uniformly bounded for $z \in
(0,1]$.
\begin{equation}\label{splitti}
    \mathcal{F}_{12} B^{-1}_{12} (z) \mathcal{C}_{12;23}
(z)\mathcal{F}_{12}^{-1} = K_1 (z)+ K_2(z),
\end{equation}
where
\begin{eqnarray}
 K_1 (z)= \frac 1{z + 1} \mathcal{F}_{12} \mathcal{C}_{12;23}
(z)\mathcal{F}_{12}^{-1} ,  \label{k1}\\
  K_2 (z)= \mathcal{F}_{12} \Bigl(B^{-1}_{12}(z) - \frac 1{z + 1}
\Bigr)\mathcal{C}_{12;23} (z)\mathcal{F}_{12}^{-1} .  \label{k2}
\end{eqnarray}
$K_1 (z)$ is uniformly norm--bounded and for $z \to +0$ converges in norm by
Lemma~\ref{lem:7,5}. Below we prove that the Hilbert-Schmidt norm of $K_2 (z)$
is bounded for $z\in (0,1]$.
Let us first consider the Fourier transformed interaction term $\mathcal{F}_{12}
V^{1/2}_{23} \mathcal{F}_{12}^{-1}$.
In  Jacobi coordinates the interaction term has the form $V^{1/2}_{23} =
V^{1/2}_{23}(\beta x + \gamma y)$,
where $\beta $ and $\gamma \neq 0$ are real constants depending on masses $\beta
= -m_2 \hbar / ((m_1 + m_2)\sqrt{2\mu_{12}})$
and $\gamma = \hbar/\sqrt{2M_{12}}$. The Fourier transformed operator acts on
$\phi(x, p_y)$ as
 \begin{equation}\label{four1}
\fl \mathcal{F}_{12} V^{1/2}_{23} \mathcal{F}_{12}^{-1} \phi = \frac 1{(2\pi)^{3/2}
\gamma^3}\int d^3 p'_y  \widehat{V^{1/2}_{23}} ((p_y - p'_y)/\gamma)
\exp{\Bigl\{i\frac{\beta}{\gamma} x \cdot (p_y - p'_y)\Bigr\}}\phi(x, p'_y),
\end{equation}
where $\widehat{V^{1/2}_{23}} \in L^2(\mathbb{R}^3)$ is a Fourier transform of
$V^{1/2}_{23} \in L^2(\mathbb{R}^3)$.
For the kernel of $K_2 (z)$ we get
 \begin{eqnarray}
K_2 (x,p_y;x',p'_y) = \frac 1{2^{7/2}\pi^{5/2} \gamma^3} \left[\frac 1{z + 1 +
t(p_y)} - \frac 1{z + 1} \right] V_{12}^{1/2} (\alpha x) \nonumber \\
\times \frac{e^{-\sqrt{p_y^2 + z^2} |x-x'|}}{|x-x'|}
\exp{\Bigl\{i\frac{\beta}{\gamma} x' \cdot (p_y - p'_y)\Bigr\}}
\widehat{V^{1/2}_{23}} ((p_y - p'_y)/\gamma) . \label{four27}
\end{eqnarray}
For the square of the Hilbert-Schmidt norm we obtain
 \begin{equation}\label{four7}
\| K_2 (z) \|^2_2 = \frac 1{2^7 \pi^5 } cc' {\tilde c} \int_{|p_y| \leq 1} d^3
p_y \; \left[\frac 1{z + \sqrt{|p_y|}} - \frac 1{z + 1} \right]^2 \frac
1{\sqrt{p_y^2 + z^2}} ,
\end{equation}
where $c, c'$ are defined in (\ref{aijacts9})--(\ref{aijacts10}) and
\begin{equation}\label{four8}
    {\tilde c} = \frac 1{\gamma^6} \int d^3 p'_y |\widehat{V^{1/2}_{23}}
(p'_y/\gamma)|^2
\end{equation}
is finite because $\widehat{V^{1/2}_{23}} \in L^2$.
Estimating the integral in (\ref{four7}) we finally obtain
 \begin{equation}\label{four9}
\| K_2 (z) \|^2_2 \leq \frac 1{2^7 \pi^5 } cc' {\tilde c} \int_{|p_y| \leq 1}
d^3 p_y \; \frac 1{p_y^2} = \frac 1{2^5 \pi^4} cc' {\tilde c}  .
\end{equation}
The strong continuity of $K_2 (z)$ for $z \to +0$ follows from the explicit form
of the kernel in (\ref{four27}).
\end{proof}

\begin{lemma}\label{lem:9}
Suppose $H(\lambda)$ defined in (\ref{xc31aa})--(\ref{:xc31aa}) satisfies R1,
R4-7. If $\psi(\lambda_{n_k})$ is a weakly convergent subsequence of
$\psi(\lambda_n)$,
then $V^{1/2}_{ij}\psi(\lambda_{n_k})$ converges in norm.
\end{lemma}
\begin{proof}
Let $J_s \in C^{2}
(\mathbb{R}^{3N-3})$ denote the Ruelle--Simon partition of unity, see Definition
3.4 and Proposition 3.5 in \cite{ims}.
For $s=1,2,3$ one has $J_s \geq 0$, $\sum_s J^{2}_s =1$ and $J_s (\lambda x) =
J_s (x)$ for $\lambda \geq 1$ and $|x|= 1$. Besides there exists $C > 0$ such
that for $i \neq s$
\begin{equation}\label{ims5}
    \supp J_s \cap \{ x | |x| > 1 \} \subset \{x|\; |r_i - r_s | \geq C |x|\} .
\end{equation}
By the IMS formula (Theorem 3.2 in \cite{ims}) the
Hamiltonian $H(\lambda)$ can be decomposed as
\begin{equation}\label{ims}
    H (\lambda) = \sum_{s=1}^3 J_s H_s (\lambda) J_s + K(\lambda) ,
\end{equation}
where
\begin{eqnarray}
\fl H_s (\lambda) = H_0 - \lambda V_{lm}, \quad \quad (l\neq s, m \neq s)\\
\fl K(\lambda) = - \lambda \sum_{s=1}^3  (V_{ls} + V_{ms} )  |J_s |^2 + \sum_{s=1}^3
 |\nabla J_s |^2 \label{K} \quad \quad (l\neq s, m \neq s , l \neq m) .
\end{eqnarray}
By the properties of $J_s$ one has $|\nabla J_s |^2 \in L^\infty_\infty
(\mathbb{R}^{3N-3})$, which makes $|\nabla J_s |^2$ relatively $H_0$--compact,
see Lemma~7.11 in \cite{teschl}.

By condition of the lemma $\psi_k \wto \psi_{cr}$, where $\psi_{cr} \in D(H_0)$
by Lemma~\ref{lem:2} and for brevity we denote
$\psi_k := \psi(\lambda_{n_k})$. We shall prove the lemma in three steps given
by the following equations
\begin{eqnarray}
(a) \quad \lim_{k \to \infty} \Bigl((\psi_k - \psi_{cr} ) , K(\lambda_{n_k})
(\psi_k - \psi_{cr} ) \Bigr) = 0 \label{kg0} \\
(b) \quad \lim_{k \to \infty} \Bigl((\psi_k - \psi_{cr} )  , H(\lambda_{n_k})
(\psi_k - \psi_{cr} )  \Bigr)= 0 \label{kg0:b}\\
(c) \quad \lim_{k \to \infty} \Bigl((\psi_k - \psi_{cr} )  , V_{ij}  (\psi_k -
\psi_{cr} )  \Bigr)= 0 . \label{kg0:c}
\end{eqnarray}
From $(c)$ the statement of the lemma clearly follows. Let us
start with $(a)$. From R6 we have
\begin{equation}\label{tobe}
 |(f,K(\lambda) f)| \leq (f,\tilde K f) \quad \quad (\forall f \in D(H_0) ) ,
\end{equation}
where the operator $\tilde K$ is defined through
\begin{eqnarray}
\tilde K = \lambda \sum_{s=1}^3  (F_{ls} + F_{ms} )  |J_s |^2 + \sum_{s=1}^3
|\nabla J_s |^2 \label{tilk1}  \quad \quad (l\neq s, m \neq s , l \neq m) .
\end{eqnarray}
$\tilde K$ is relatively $H_0$--compact by Lemma~\ref{lem:5} and thus by
Lemma~\ref{lem:2}
\begin{equation}\label{j233}
((\psi_k - \psi_{cr} ) , \tilde K (\psi_k - \psi_{cr} ) ) \to 0
\end{equation}
This proves (a). Rewriting the expression in (b) we obtain
\begin{eqnarray}
\bigl( (\psi_k - \psi_{cr} ), H(\lambda_{n_k}) \: (\psi_k - \psi_{cr} )\bigr)
= E(\lambda_{n_k}) \bigl((\psi_k - \psi_{cr} ), \psi_k\bigr) \nonumber\\
- \bigl((\psi_k -\psi_{cr} ), H(\lambda_{cr})\psi_{cr} \bigr)  - [\lambda_{n_k} -
\lambda_{cr}]\bigl((\psi_k -\psi_{cr} ), V\psi_{cr} \bigr) \label{tobe34} ,
\end{eqnarray}
where we have used $H(\lambda_{n_k}) = H(\lambda_{cr})
+ [\lambda_{n_k} - \lambda_{cr}]V$. All terms on the rhs of
(\ref{tobe34}) go to zero because $E(\lambda_{n_k}) \to 0$ and
$\psi_k \wto \psi_{cr}$. It remains to be shown that $(c)$ is true.
\begin{eqnarray}
\fl \lim_{k \to \infty} \Bigl((\psi_k - \psi_{cr} )  , V_{ij}  (\psi_k - \psi_{cr} )
 \Bigr)= \sum_{s=1}^3 \lim_{k \to \infty} \Bigl((\psi_k - \psi_{cr} )  , J_s
V_{ij} J_s (\psi_k - \psi_{cr} )  \Bigr) \nonumber\\
\fl = \lim_{k \to \infty} \Bigl((\psi_k - \psi_{cr} )  , J_l V_{ij} J_l (\psi_k -
\psi_{cr} )  \Bigr) \quad \quad (l \neq i \neq j) \label{sima2},
\end{eqnarray}
where we have used that $J_i V_{ij}$ and $J_j V_{ij}$ are relatively
$H_0$--compact by Lemma~\ref{lem:5} and the corresponding scalar products vanish
by Lemma~\ref{lem:2}.

From (a), (b) and (\ref{ims}) we obtain
\begin{equation}\label{tobe4}
\bigl((\psi_k - \psi_{cr} ), J_l  H_{l} (\lambda_{n_k}) J_l (\psi_k - \psi_{cr}
)  \bigr)
\to 0 \quad \quad (\forall l) .
\end{equation}
Together with R7 this gives us
\begin{equation}\label{kotz1}
\lim_{k \to \infty} \left((\psi_k - \psi_{cr} )  , J_l V_{ij} J_l (\psi_k -
\psi_{cr} )  \right) = 0\quad \quad (l \neq i \neq j) .
\end{equation}
Finally, comparing (\ref{kotz1}) and (\ref{sima2}) we conclude
that (c) holds. \end{proof}

\begin{proof}[Proof of Theorem~\ref{th:2}]
It is enough to show that any weakly converging subsequence of $\psi(\lambda_n)$
converges in norm.
Indeed, in this case $\psi(\lambda_n)$ does not spread by Lemma~\ref{lem:4}  and
thus by Theorem~\ref{th:1}
there must exist a bound state at threshold. Suppose $\psi(\lambda_{n_s}) $ is a
weakly converging subsequence,
that is  $\psi(\lambda_{n_s}) \wto \psi_{cr}$ and we must prove $\|
\psi(\lambda_{n_s}) -\psi_{cr}\| \to 0$.

By Schr\"odinger equation for $k_{n_s}^2 = -E_{n_s} >0$
\begin{equation}\label{nnr1}
 \fl   \psi(\lambda_{n_s}) = \lambda_{n_s} \sum_{i<j} [H_0 + k_{n_s}^2]^{-1} V_{ij}
\psi(\lambda_{n_s}) = \lambda_{n_s} \sum_{i<j} \mathcal{A}_{ij} (k_{n_s})
\bigl[ B^{-1}_{ij} (k_{n_s}) V^{1/2}_{ij}\psi(\lambda_{n_s}) \bigr] ,
\end{equation}
where $\mathcal{A}_{ij}$ is defined in (\ref{ya23}). By Lemma~\ref{lem:6}
$\psi(\lambda_{n_s})$ converges in norm if the sequence $B^{-1}_{ij} (k_{n_s})
V^{1/2}_{ij}\psi(\lambda_{n_s}) $ does. The convergence of the latter we prove
below. From (\ref{nnr1}) we obtain
\begin{equation}\label{nnr2}
    V^{1/2}_{ij} \psi(\lambda_{n_s}) = \lambda_{n_s} \sum_{l < m}
\mathcal{C}_{ij;lm} (k_{n_s}) [V^{1/2}_{lm} \psi(\lambda_{n_s}) ] .
\end{equation}
Using (\ref{ay2222}) we rewrite (\ref{nnr2})
\begin{equation}\label{nnr4}
    V^{1/2}_{ij} \psi(\lambda_{n_s}) = \lambda_{n_s} \mathcal{R}_{ij} (k_{n_s})
\sum_{l < m \\ (lm) \neq (ij)} \mathcal{C}_{ij;lm}(k_{n_s}) (V^{1/2}_{lm}
\psi(\lambda_{n_s}) ) .
\end{equation}
Now we act with $B^{-1}_{ij}(k_{n_s})$ on both parts of (\ref{nnr4}) and use
that it commutes with $\mathcal{R}_{ij} (k_{n_s})$
\begin{equation}\label{nnr5}
 \fl   B^{-1}_{ij}(k_{n_s}) V^{1/2}_{ij} \psi(\lambda_{n_s}) = \lambda_{n_s}
\mathcal{R}_{ij} (k_{n_s}) \sum_{l < m \\
    (lm) \neq (ij)} B^{-1}_{ij}(k_{n_s}) \mathcal{C}_{ij;lm}(k_{n_s})
\bigl(V^{1/2}_{lm} \psi(\lambda_{n_s}) \bigr)   .
\end{equation}
By Lemmas~\ref{lem:7},\ref{lem:8},\ref{lem:9} the rhs converges in
norm.\end{proof}

In the next sections our aim is to analyse the case when one pair of particles
has a zero--energy resonance. We would show (Theorem~\ref{Q@th:1}) that in this
case Theorem~\ref{th:2} does not generally hold. 
This shows that the condition of Theorem~\ref{th:2} on the absence of resonances
in particle pairs is essential.

\section{A Zero Energy Resonance in a 2--Particle System}\label{Q@sec:2}

In this section we shall use the method of \cite{klaus1} to prove a result
similar to Lemma~2.2 in \cite{sobol}. Let us consider the Hamiltonian of 2
particles in $\mathbb{R}^3$
\begin{equation}\label{Q@2ps}
h_{12}(\varepsilon) := - \Delta_x - (1+\varepsilon) V_{12} (\alpha x) ,
\end{equation}
where $\varepsilon \geq 0$ is a parameter and $\alpha$ is defined right after
(\ref{redew7}). Additionally, we require
\begin{list}{$\overline{R}$\arabic{foo}}
{\usecounter{foo}
    \setlength{\rightmargin}{\leftmargin}}
\item
$0\leq V_{12} (\alpha x) \leq b_1 e^{-b_2 |x|}$, where $b_{1,2} >0$ are
constants.
\item
$h_{12}(0) \geq 0$ and $\sigma (h_{12}(\varepsilon)) \cap (-\infty, 0) \neq
\emptyset$ for $\varepsilon >0$. \end{list}
The requirement $\overline{R}$2 means that $h_{12}(0)$ has a resonance at zero
energy, that is, negative energy bound states emerge iff the coupling constant
is incremented by an arbitrary amount 
(in terminology of \cite{klaus1} the system is at the  coupling constant
threshold).

The following integral operator appears in the Birman--Schwinger principle
\cite{reed,klaus1}
\begin{equation}\label{Q@bms}
    L(k) := \sqrt{V_{12}} \Bigl(- \Delta_x + k^2 \Bigr)^{-1} \sqrt{V_{12}}.
\end{equation}
$L(k)$ is analytic for $\real k >0$. Due to $\overline{R}$1 one can use the
integral representation and analytically continue $L(k)$ into the interior of
the disk on the complex plane, which has its centre at $k=0$ and the 
radius $|b_2|$ \cite{klaus1}. The analytic continuation is denoted as $\tilde
L(k) = \sum_n \tilde L_n k^n$, where $\tilde L_n$ are Hilbert-Schmidt operators.
\begin{remark}
In Sec.~2 in \cite{klaus1} (page 255) Klaus and Simon consider only finite range
potentials. In this case $L(k)$ can be analytically continued into the whole
complex plane. As the authors mention it in Sec.~9 the case of potentials with
an exponential fall off requires only a minor change: $L(k)$ extends
analytically as a bounded operator to the domain $\{k|\: \real k > -b_2 \}$.
\end{remark}

Under requirements $\overline{R}$1, $\overline{R}$2 the operator $L(0) = \tilde
L(0)$ is Hilbert-Schmidt and its maximal eigenvalue is equal to one
\begin{equation}\label{refsugg14}
L(0) \phi_0 = \phi_0 .
\end{equation}
$L(0)$ is positivity--preserving, hence, the maximal eigenvalue is
non--degenerate and $\phi_0 \geq 0$. We choose the normalization, where
$\|\phi_0 \| = 1$.

By the standard Kato--Rellich perturbation theory \cite{kato,reed} there exists
$\rho > 0$ such that for $|k| \leq \rho$
\begin{equation}
\tilde L(k) \phi(k)= \mu(k)\phi(k) ,
\end{equation}
where $\mu(k), \phi(k)$ are analytic, $\mu(0) = 1$, $\phi(0)=\phi_0$ and the
eigenvalue $\mu(k)$ is non--degenerate.
By Theorem 2.2 in \cite{klaus1}
\begin{equation}\label{Q@muexp}
\mu (k) = 1 - a k + O(k^2) ,
\end{equation}
where
\begin{equation}\label{Q@whatisa}
a = (\phi_0, (V_{12})^{1/2})^2/(4\pi) > 0 .
\end{equation}
The orthonormal projection operators
\begin{eqnarray}
\mathbb{P}(k) := (\phi(k), \cdot)\phi(k) = (\phi_0, \cdot)\phi_0 +
\mathcal{O}(k) , \label{Q@prp}\\
\mathbb{Q}(k) := 1 - \mathbb{P}(k) \label{Q@prq}
\end{eqnarray}
are analytic for $|k| < \rho$ as well. Our aim is to analyse the following
operator function on $k \in (0, \infty)$
\begin{equation}\label{Q@defw}
W(k) = [1-L(k)]^{-1} .
\end{equation}
By the Birman--Schwinger principle $\| L(k) \| < 1$ for $k >0$, which makes
$W(k)$ well--defined.
\begin{lemma}\label{Q@lem:1}
There exists $0< \rho_0 < 1$ such that for $k \in (0,\rho_0)$
\begin{equation}\label{Q@expans}
W(k) = \frac{\mathbb{P}_0}{ak} + \mathcal{Z}(k) ,
\end{equation}
where $\mathbb{P}_0 := (\phi_0 , \cdot)\phi_0 $ and $\sup_{k \in (0,\rho_0)}
\|\mathcal{Z}(k)\| < \infty$.
\end{lemma}
\begin{proof}
$\tilde L(k) = L(k)$ when $k \in (0,\rho)$. We get from (\ref{Q@defw})
\begin{eqnarray}
W(k) = [1-L(k)]^{-1} = [1-L(k)]^{-1} \mathbb{P}(k)  +  [1-L(k)]^{-1}
\mathbb{Q}(k) \nonumber\\
= [1-\mu(k) \mathbb{P}(k) ]^{-1} \mathbb{P}(k) + [1-\mathbb{Q}(k) L(k)]^{-1}
\mathbb{Q}(k) \nonumber\\
= \frac 1{1-\mu(k)} \mathbb{P}(k) + \mathcal{Z}'(k),
\end{eqnarray}
where
\begin{equation}
\mathcal{Z}'(k) := [1-\mathbb{Q}(k) L(k)]^{-1} \mathbb{Q}(k) .
\end{equation}
Note that $\sup_{k \in (0,\rho)}\| \mathbb{Q}(k) L(k)\| < 1$  because the
eigenvalue $\mu(k)$ remains isolated for $k \in [0,\rho)$. Thus $\mathcal{
Z}'(k) = \mathcal{O}(1)$. Using (\ref{Q@muexp}),(\ref{Q@whatisa}) and (\ref{Q@prp}) proves
the lemma. Clearly, one can always choose $\rho_0 < 1$. \end{proof}

\begin{remark}
The singularity of $W(k)$ near $k=0$ has been analysed in \cite{sobol} (Lemma
2.2 in \cite{sobol}), see also \cite{yafaev}). 
The decomposition (\ref{Q@expans}) differs in the sense that $\mathcal{Z}(k)$ is
uniformly bounded in the vicinity of $k=0$. 
The price we paid for it is the requirement $\overline{R}$1 on the exponential
fall off of $V_{12}$.
\end{remark}

\section{Zero Energy Resonance in a 3--Particle system}\label{Q@sec:3}

Let us consider the Schr\"odinger operator for three particles in $\mathbb{R}^3$
\begin{equation}\label{Q@xc31aa}
    H = H_0 - V_{12}(r_1 - r_2) - V_{13}(r_1 - r_3) - V_{23}(r_2 - r_3),
\end{equation}
where $r_i$ are particle position vectors and $H_0$ is the kinetic energy
operator with the centre of mass removed.
Apart from $\overline{R}$1, $\overline{R}$2 we shall need the following
additional requirement
\begin{list}{$\overline{R}$\arabic{foo}}
{\usecounter{foo}
    \setlength{\rightmargin}{\leftmargin}}
\setcounter{foo}{2}
\item$V_{13}, V_{23} \in L^2 (\mathbb{R}^3) + L^\infty_\infty (\mathbb{R}^3) $
 and  $V_{13}, V_{23} \geq 0$ and $V_{23} \neq 0$.
\end{list}
Here we shall prove
\begin{theorem}\label{Q@th:1}
Suppose $H$ defined in (\ref{Q@xc31aa}) satisfies $\overline{R}$1,
$\overline{R}$2, $\overline{R}$3. Suppose additionally that $H \geq 0$ and
$H\psi_0 =0$, where $\psi_0 \in D(H_0)$. Then $\psi_0 = 0$.
\end{theorem}
We defer the proof to the end of this section. Our next aim is to derive the
inequality (\ref{Q@e11})-(\ref{Q@e12a}).

We use the same Jacobi coordinates as in Sec.~\ref{sec:3} so that (\ref{ay4})
holds.
The full set of coordinates in $\mathbb{R}^6$ is labelled by $\xi$. We shall need
the following trivial technical lemmas.
\begin{lemma}\label{Q@lem:2}
Suppose an operator $A$ is positivity preserving and $\| A\| < 1$. Then 
$(1-A)^{-1}$ is bounded and positivity preserving.
\end{lemma}
\begin{proof}
A simple expansion of $(1-A)^{-1}$ into von Neumann series. \end{proof}
\begin{lemma}\label{Q@lem:3}
Suppose $g \in L^1 (\mathbb{R}^3) $, $\|g\|_1 >0$  and $g(y) \geq 0$. Then for
all $\epsilon_0 > 0$
\begin{equation}\label{Q@sq1}
    \lim_{z \to + 0} \int_{|p_y|\leq \epsilon_0} d^3 p_y \frac{|\hat
g|^2}{(p_y^2 + z^2)^{3/2}} = \infty
\end{equation}
\end{lemma}
\begin{proof}
Let us set
\begin{equation}\label{Q@sq3}
    J_{\epsilon}(z) = \int_{|p_y| \leq \epsilon} d^3 p_y \frac{1}{(p_y^2 +
z^2)^{3/2}} \left| \int d^3 y e^{i p_y \cdot y} g(y) \right|^2 .
\end{equation}
We have
\begin{equation}\label{Q@sq4}
    J_{\epsilon} (z) \geq \int_{|p_y| \leq \epsilon} d^3 p_y \frac{1}{(p_y^2 +
z^2)^{3/2}} \left| \int d^3 y \;  g(y) \cos{(p_y \cdot y)}\right|^2 .
\end{equation}
Let us fix $r$ so that
\begin{equation}\label{Q@sq4a}
    \int_{|y|>r} d^3 y g(y) = \frac 14 \|g\|_1
\end{equation}
Setting $\epsilon = \min [\epsilon_0 , \pi/(3r)]$ we get
\begin{equation}\label{Q@sq5}
\cos{(p_y \cdot y)} \geq \frac 12  \quad \quad \mathrm{if} \quad |p_y| \leq
\epsilon , |y|\leq r .
\end{equation}
Substituting (\ref{Q@sq5}) and (\ref{Q@sq4a}) into (\ref{Q@sq4}) we get
\begin{equation}\label{Q@sq6}
    J_{\epsilon_0} (z) \geq J_\epsilon (z) \geq \frac{\| g \|_1^2}{64}
\int_{|p_y| \leq \epsilon} d^3 p_y \frac{1}{(p_y^2 + z^2)^{3/2}} .
\end{equation}
The integral in (\ref{Q@sq6}) logarithmically diverges for $z \to +0$. 
\end{proof}

So let us assume that there is a bound state $\psi_0 \in D(H_0)$ at zero energy,
where $\psi_0 >0$ because it is the ground state, see \cite{reed} Sec.~XIII.12.
Then we would have
\begin{equation}\label{Q@e1}
    H_0 \psi_0 = V_{12} \psi_0 + V_{13} \psi_0 + V_{23} \psi_0 ,
\end{equation}
Adding the term $z^2 \psi_0$ (where here and further $z > 0$ ) and acting with
an inverse operator on both sides of (\ref{Q@e1}) gives
\begin{eqnarray}
    \psi_0 = [H_0 + z^2]^{-1}V_{12} \psi_0 + [H_0 + z^2]^{-1} V_{13} \psi_0 +
[H_0 + z^2]^{-1}V_{23} \psi_0 \nonumber\\
    + z^2 [H_0 + z^2]^{-1}\psi_0 . \label{Q@e3a}
\end{eqnarray}
From now we let $z$ vary in the interval $(0,\rho_0/2)$, where $\rho_0 < 1$ was
defined in Lemma~\ref{Q@lem:1}. The operator $[H_0 + z^2]^{-1}$ is positivity
preserving, see, 
for example, \cite{reed}, Example 3 from Sec.~IX.7 in vol. 2 and Theorem~XIII.44
in vol. 4. Thus we obtain the inequality
\begin{equation}\label{Q@e4}
    \psi_0 \geq [H_0 + z^2]^{-1}\sqrt{V_{12}} (\sqrt{V_{12}} \psi_0 )
\end{equation}
Now let us focus on the term $\sqrt{V_{12}} \psi_0 $. Using (\ref{Q@e3a}) we get
\begin{eqnarray}\label{Q@e5}
    \Bigl[1- \sqrt{V_{12}} (H_0 + z^2)^{-1} \sqrt{V_{12}}\Bigr]  \sqrt{V_{12}}
\psi_0 =  \sqrt{V_{12}} [H_0 + z^2]^{-1} V_{13} \psi_0 \nonumber\\
    +  \sqrt{V_{12}} [H_0 + z^2]^{-1}V_{23} \psi_0 + z^2  \sqrt{V_{12}} [H_0 +
z^2]^{-1}\psi_0
\end{eqnarray}
And by Lemma~\ref{Q@lem:2}
\begin{equation}\label{Q@e6}
    \sqrt{V_{12}} \psi_0 \geq  \Bigl[1- \sqrt{V_{12}} (H_0 + z^2)^{-1}
\sqrt{V_{12}}\Bigr]^{-1}  \sqrt{V_{12}} [H_0 + z^2]^{-1} V_{23} \psi_0  .
\end{equation}
The resolvent identity reads
\begin{equation}\label{Q@e668}
[H_0 + z^2]^{-1} - [H_0 + 1]^{-1} = (1-z^2)[H_0 + 1]^{-1} [H_0 + z^2]^{-1} .
\end{equation}
Clearly, for $z \in (0,1)$ the difference on the lhs of (\ref{Q@e668}) is a
positivity preserving operator.
Using this fact we can transform (\ref{Q@e6}) into
\begin{equation}\label{Q@e669}
    \sqrt{V_{12}} \psi_0 \geq  \Bigl[1- \sqrt{V_{12}} (H_0 + z^2)^{-1}
\sqrt{V_{12}}\Bigr]^{-1}  \sqrt{V_{12}} [H_0 + 1]^{-1} V_{23} \psi_0  .
\end{equation}

It is technically convenient to cut off the wave function $\psi_0$ by
introducing
\begin{equation}\label{Q@e7}
\psi_1 (\xi) := \psi_0 (\xi) \chi_{\{\xi|\; |\xi| \leq b\}} ,
\end{equation}
where clearly $\psi_1 \in L^2 \cap L^1 (\mathbb{R}^6)$ and $b>0$ is fixed so
that $\|V_{23} \psi_1\| \neq 0$ (which is always possible since $V_{23} \neq
0$).

Applying again Lemma~\ref{Q@lem:2} we get out of (\ref{Q@e669})
\begin{equation}\label{Q@e6new}
    \sqrt{V_{12}} \psi_0 \geq  \Bigl[1- \sqrt{V_{12}} (H_0 + z^2)^{-1}
\sqrt{V_{12}}\Bigr]^{-1}  \sqrt{V_{12}} [H_0 + 1]^{-1} V_{23} \psi_1 \: .
\end{equation}
Substituting (\ref{Q@e6new}) into (\ref{Q@e4}) gives that for all $z \in (0,
\rho_0 /2)$
\begin{equation}\label{Q@e11}
\psi_0 \geq f(z) \geq 0,
\end{equation}
where
\begin{eqnarray}
    f(z) = [H_0 + z^2]^{-1}\sqrt{V_{12}} \Bigl[1- \sqrt{V_{12}} (H_0 + z^2)^{-1}
\sqrt{V_{12}}\Bigr]^{-1}  \nonumber \\
    \times \sqrt{V_{12}} [H_0 + 1]^{-1} V_{23} \psi_1 \: . \label{Q@e12a}
\end{eqnarray}

Our aim is to prove that $\lim_{z \to +0} \|f(z)\| = \infty$, which would be in
contradiction with (\ref{Q@e11}) because $\psi_0$ is the normalized ground state wave funtion. 
Let us define
\begin{eqnarray}
\Phi (x,y) := [H_0 + 1]^{-1} V_{23} \psi_1 , \label{Q@gaze3}\\
    g(y) := \int d x \; \Phi (x,y) \sqrt{V_{12} (\alpha x)}\phi_0 (x) ,
\label{Q@gaze4}
\end{eqnarray}
where $\phi_0$ is defined in (\ref{refsugg14}). 
\begin{lemma}\label{Q@lem:4}
$g \in L^1 \cap L^2 (\mathbb{R}^3)$ and $\|g\|_1 >0$.
\end{lemma}
\begin{proof}
Following \cite{klaus1} let us denote by $G_0 (\xi - \xi', 1)$ the integral
kernel of $[H_0 +1]^{-1}$. 
We need a rough upper bound on $G_0 (\xi,1)$. Using the formula on p. 262 in
\cite{klaus1} we get
\begin{eqnarray}
(4\pi)^3 |\xi|^4 e^{|\xi|/2} G_0 (\xi , 1) = \int_o^\infty t^{-3} e^{|\xi|/2}
e^{-t|\xi|^2} e^{-1/(4t)} dt \nonumber\\
\leq \int_0^\infty t^{-3} e^{-3/(16t)} dt = \frac{256}9
\end{eqnarray}
Hence,
\begin{equation}\label{Q@bgf}
G_0 (\xi,1) \leq \frac 4{9\pi |\xi|^4}  e^{-|\xi|/2}  .
\end{equation}
Using $\|\sqrt{V_{12}}\phi_0\|_\infty < \infty $ we get $g \in L^1 \cap L^2
(\mathbb{R}^3)$ if $\Phi \in L^1 \cap L^2 (\mathbb{R}^6)$. Because $\Phi \in L^2
(\mathbb{R}^6)$ to prove $\Phi \in L^1 (\mathbb{R}^6)$ it suffices to show that
$\chi_{\{\xi|\; |\xi| \geq 2b\}}\Phi(\xi) \in L^1 (\mathbb{R}^6)$, where $b$ was
defined after Eq.~(\ref{Q@e7}). This follows from (\ref{Q@bgf})
\begin{eqnarray}\label{Q@bgfa}
\chi_{\{\xi|\; |\xi| \geq 2b\}}\Phi(\xi) \leq \chi_{\{\xi|\; |\xi| \geq 2b\}}
\int_{|\xi'|\leq b} d^6 \xi' G_0 (\xi - \xi',1) V_{23}\psi_1(\xi') \nonumber\\
\leq   \chi_{\{\xi|\; |\xi| \geq 2b\}} \frac 4{9\pi (|\xi| - b)^4}  e^{-(|\xi| -
b)/2} \bigl\|V_{23}\psi_1\bigr\|_1  \in L^1 (\mathbb{R}^6)
\end{eqnarray}
From $\Phi (x,y) >0$ it follows that $\|g\|_1 > 0$. \end{proof}

Applying $\mathcal{F}_{12}$ to (\ref{Q@e12a}) we get
\begin{eqnarray}
\fl     \hat f (z) = [-\Delta_x +p_y^2 + z^2]^{-1}\sqrt{V_{12}} \Bigl[1-
\sqrt{V_{12}} (-\Delta_x +p_y^2 + z^2)^{-1} \sqrt{V_{12}}\Bigr]^{-1}  
\nonumber\\
    \sqrt{V_{12}} [-\Delta_x + p_y^2 + 1]^{-1} \widehat{V_{23} \psi_1} .
\label{Q@gaze2}
\end{eqnarray}

From now on $z\in (0,\rho_0/2)$. By Lemma~\ref{Q@lem:1} for $|p_y | < \rho_0 /2$
and $z < \rho_0 /2$
\begin{equation}\label{Q@sobef}
  \fl  \Bigl[1- \sqrt{V_{12}} \Bigl(-\Delta_x +p_y^2 + z^2\Bigr)^{-1}
\sqrt{V_{12}}\Bigr]^{-1} = \frac{\mathbb{P}_0}{a \sqrt{p_y^2 + z^2}} +
\mathcal{Z}\Bigl(\sqrt{p_y^2 + z^2}\Bigr)  ,
\end{equation}
where $a$ and $\phi_0(x)$ are defined in Sec.~\ref{Q@sec:2} and $\mathbb{P}_0$
acts on $u(x,p_y)$ as $\mathbb{P}_0 u(x,p_y) = \phi_0 (x) \int \phi_0 (x')
u(x',p_y)\: d x' $. Substituting (\ref{Q@sobef}) into
(\ref{Q@gaze2}) and denoting for brevity $\chi_0 (p_y) :=
\chi_{\{p_y | \; |p_y| < \rho_0 /2\}}$ we obtain
\begin{equation}\label{Q@e13}
    \chi_0 (p_y) \hat f(z) = \hat f_1 (z) + \hat f_2 (z),
\end{equation}
where
\begin{eqnarray}
    \hat f_1(z) = \chi_0 (p_y) \frac{\hat g(p_y)}{\sqrt{p_y^2 + z^2}} [-\Delta_x
+p_y^2 + z^2]^{-1}\bigl(\sqrt{V_{12}} \phi_0 (x)\bigr) , \label{Q@f1}\\
    \hat f_2(z) = \chi_0 (p_y) [-\Delta_x +p_y^2 + z^2]^{-1}\sqrt{V_{12}}
\mathcal{Z}\Bigl(\sqrt{p_y^2 + z^2}\Bigr) \nonumber\\
    \sqrt{V_{12}} [-\Delta_x + p_y^2 + 1]^{-1} \bigl(\mathcal{F}_{12}V_{23}
\mathcal{F}^{-1}_{12}\bigr) \hat \psi_1  \label{Q@f2aa} ,
\end{eqnarray}
and we have used (\ref{Q@gaze3})-(\ref{Q@gaze4}).
The next lemma follows from the results of Sec.~\ref{sec:3}.
\begin{lemma}\label{Q@lem:5}
$\sup_{z \in (0,\rho_0/2)}\|f_2 (z)\| < \infty$
\end{lemma}
\begin{proof}
Let us rewrite (\ref{Q@f2aa}) in the form
\begin{equation}\label{Q@f2a}
f_2 (z) = \mathcal{A}(z) \mathcal{B}(z) \mathcal{C}(z) \psi_0 ,
\end{equation}
where
\begin{eqnarray}
\mathcal{A}(z) = \chi_0 (p_y) [-\Delta_x +p_y^2 + z^2]^{-1}\sqrt{V_{12}}
[1+t(p_y)+z]  , \\
\mathcal{B}(z) =  \chi_0 (p_y) \mathcal{Z}\Bigl(\sqrt{p_y^2 + z^2}\Bigr) , \\
\mathcal{C}(z) =   \chi_0 (p_y) \sqrt{V_{12}} [-\Delta_x + p_y^2 + 1]^{-1}
[1+t(p_y)+z]^{-1} \bigl(\mathcal{F}_{12}V_{23} \mathcal{F}^{-1}_{12}\bigr) ,
\end{eqnarray}
and  $t(p_y)$ is defined as in (\ref{ay669}).
Note that by (\ref{Q@expans}) $\mathcal{Z} \Bigl(\sqrt{p_y^2 + z^2}\Bigr)$ is a
difference of two operators each of which commutes with the operator of
multiplication by $[1+t(p_y)+z]$. We need to show that each of the three
operators in the product in (\ref{Q@f2a}) are uniformly norm--bounded for $z \in
(0,\rho_0/2)$. That $\sup_{z \in (0,\rho_0/2)}\|\mathcal{B} (z)\| < \infty$
follows from Lemma~\ref{Q@lem:1}. That $\sup_{z \in (0,\rho_0/2)}\|\mathcal{A}
(z)\|, \|\mathcal{C} (z)\| < \infty$ follows from the proofs of Lemmas
\ref{lem:6}, \ref{lem:8} in Sec.~\ref{sec:3}. Let us, however, repeat the
argument here. Taking into account that $0 < z < \rho_0/2 < 1$ we obtain
\begin{eqnarray}
\fl \|\mathcal{A} (z)\| = \Bigl\| \chi_0 (p_y) \bigl[-\Delta_x +p_y^2 +
z^2\bigr]^{-1}\sqrt{V_{12}} \bigl[1+t(p_y)+z\bigr] \Bigr\| \nonumber \\
\fl \leq \Bigl\| \chi_0 (p_y) \bigl[-\Delta_x +p_y^2 + z^2\bigr]^{-1}\sqrt{V_{12}}
\sqrt{|p_y|} \Bigr\| +
z \Bigl\| \chi_0 (p_y) \bigl[-\Delta_x +p_y^2 + z^2\bigr]^{-1}\sqrt{V_{12}}
\Bigr\| \nonumber\\
\fl \leq \Bigl\| \chi_0 (p_y) \bigl[-\Delta_x +p_y^2 + z^2\bigr]^{-1}\sqrt{V_{12}}
\sqrt{|p_y|} \Bigr\| +
z\Bigl\| \bigl[-\Delta_x +z^2\bigr]^{-1}\sqrt{V_{12}} \Bigr\|
\label{Q@secondnorm}
\end{eqnarray}
It is trivial to estimate the squares of the norms on the rhs if one uses the
explicit expressions for the operator kernels. For example,
\begin{eqnarray}
\Bigl\| \chi_0 (p_y) \bigl[-\Delta_x +p_y^2 + z^2\bigr]^{-1}\sqrt{V_{12}}
\sqrt{|p_y|} \Bigr\|^2 \nonumber \\
\leq \frac 1{(4\pi)^2} \sup_{|p_y| < \rho_0 /2} |p_y|\int \int
\frac{e^{-2|p_y||x-x'|} V_{12}(\alpha x')}{|x-x'|^2} d^3 x d^3 x' = \frac {c
c'}{4\pi}  < \infty ,
\end{eqnarray}
where $c, c'$ are defined in (\ref{aijacts9})--(\ref{aijacts10}).
The second norm in (\ref{Q@secondnorm}) is estimated similarly and the result is
that $\mathcal{A} (z)$ is uniformly norm--bounded for $z \in (0,\rho_0/2)$.
Using (\ref{four1}) we can write the integral kernel of $\mathcal{C} (z)$ as
\begin{eqnarray}
\mathcal{C}(z) (x,p_y;x',p'_y) = \frac {\chi_0 (p_y) }{2^{7/2}\pi^{5/2}
\gamma^3} \left[z+\sqrt{|p_y|}\right]^{-1} V_{12}^{1/2} (\alpha x)
\nonumber \\
\times \frac{e^{-\sqrt{p_y^2 + z^2} |x-x'|}}{|x-x'|} 
\exp{\Bigl\{i\frac{\beta}{\gamma} x' \cdot (p_y - p'_y)\Bigr\}}
\widehat{V^{1/2}_{23}} ((p_y - p'_y)/\gamma) . \label{Q@four27}
\end{eqnarray}
Estimating $\| \mathcal{C} (z)\|^2$ through the square of the Hilbert--Schmidt
norm results in
\begin{equation}\label{Q@zhirn1}
\| \mathcal{C} (z)\|^2 \leq \frac{c c' \tilde c}{2^{7}\pi^{5} } \int_{|p_y| \leq
\rho_0/2} \frac{d^3 p_y}{|p_y| (z + \sqrt{|p_y|})^2} ,
\end{equation}
where $\tilde c$ is defined in (\ref{four8}). From (\ref{Q@zhirn1}) it follows
that $\sup_{z \in (0,\rho_0/2)}\|\mathcal{C} (z)\| < \infty$. \end{proof}

The last Lemma needed for the proof of Theorem~\ref{Q@th:1} is
\begin{lemma}\label{Q@lem:6}
$\lim_{z \to 0}\|f_1 (z)\| = \infty$.
\end{lemma}
\begin{proof}
We get
\begin{eqnarray}
    \fl \| \hat f_1(z) \|^2 = \frac 1{4\pi^2} \int_{|p_y| \leq \rho_0/2} d p_y
\frac{|\hat g(p_y)|^2}{p_y^2 + z^2} \int dx \int dx' \int dx'' \frac
{e^{-\sqrt{p_y^2 + z^2}|x-x'|}}{|x-x'|} \nonumber \\
 \fl   \times \frac {e^{-\sqrt{p_y^2 + z^2}|x-x''|}}{|x-x''|}
\bigl(\sqrt{V_{12}}(\alpha x') \phi_0 (x')\bigr) \bigl(\sqrt{V_{12}}(\alpha x'')
\phi_0 (x'')\bigr)  .
\end{eqnarray}
The are constants $R_0, C_0 >0$ such that
\begin{equation}\label{Q@expin}
\int d^3 x' \frac{e^{-\delta|x-x'|}}{|x-x'|} \sqrt{V_{12}(\alpha x')}\phi_0 (x')
\geq C_0 \frac{e^{-2\delta|x|}}{|x|}\chi_{\{x|\; |x|\geq R_0\}}
\end{equation}
for all $\delta >0$. Indeed, the following inequality holds for all $R_0 >0$
\begin{equation}\label{Q@zabyv}
\chi_{\{x|\; |x|\geq R_0\}} \frac{e^{-\delta|x-x'|}}{|x-x'|} \chi_{\{x' |\;
|x'|\leq R_0\}} \geq \frac{e^{-2\delta|x|}}{2|x|}\chi_{\{x|\; |x|\geq R_0\}}  .
\end{equation}
Substituting (\ref{Q@zabyv}) into the lhs of (\ref{Q@expin}) we obtain
(\ref{Q@expin}), where
\begin{equation}
C_0 = \frac 12 \int_{|x'|\leq R_0} d^3 x' \sqrt{V_{12}(\alpha x')}\phi_0 (x')
\end{equation}
and one can always choose $R_0$ so that $C_0 >0$. Using (\ref{Q@expin}) we get
\begin{equation}
 \| \hat f_1(z) \|^2  \geq c \int_{|p_y| \leq \frac{\rho_0}2} d p_y \frac{|\hat
g(p_y)|^2}{(p_y^2 + z^2)^{3/2}} ,
\end{equation}
where $c>0$ is a constant. Now the result follows from Lemmas~\ref{Q@lem:3},
\ref{Q@lem:4}.  \end{proof}

The proof of Theorem~\ref{Q@th:1} is now trivial.
\begin{proof}[Proof of Theorem~\ref{Q@th:1}]
A bound state at threshold should it exist must satisfy inequality (\ref{Q@e11})
for all $z \in (0,\rho_0/2)$. Thus $\|f(z)\|$ and, hence, $\|\chi_0 \hat f(z)\|$
are uniformly bounded for $z \in (0,\rho_0/2)$. By (\ref{Q@e13}) and
Lemmas~\ref{Q@lem:5}, \ref{Q@lem:6} this leads to a contradiction.  \end{proof}

\section{Example of a Three--Particle Zero Energy Resonance and Physical Applications}\label{Q@sec:4}

Suppose that $\overline{R}$2 is fulfilled. Let us rewrite (\ref{Q@xc31aa}) using additional coupling
constants $\Theta , \Lambda >0$ 
\begin{equation}
H(\Theta, \Lambda) = [-\Delta_x - V_{12}] -\Delta_y - \Theta V_{13} - \Lambda
V_{23}  .
\end{equation}
For simplicity, let us require that $V_{ik} \geq 0$ and $V_{ik} \in C^\infty_0
(\mathbb{R}^3)$. Let $\Theta_{cr}, \Lambda_{cr}$ denote the 2--particle coupling
constant thresholds for particle pairs 1,3 and 2,3 respectively. On one hand,
using a variational argument it is easy to show that there exists $\epsilon >0$
such that $H(\Theta, \Lambda) > 0$ if $\Theta, \Lambda \in [0,\epsilon]$ (that
is in this range $H(\Theta, \Lambda)$ has neither negative energy bound states
nor a zero energy resonance) 
\cite{gridnev,martin}. On the other hand, from the Yafaev's rigorous proof of the existence of the Efimov effect \cite{yafaev} we know 
that $H(\Theta_{cr}, \Lambda)$ has an infinite number of negative energy bound states for $\Lambda \in [0,\Lambda_{cr})$ because 
in this case two of the binary subsystems have zero energy resonances. 
So let us fix $\Lambda = \epsilon$ and let $\Theta$ vary in
the range $[\epsilon, \Theta_{cr}]$. The energy of the ground state
$E_{gr}(\Theta) = \inf \sigma \Bigl(H(\Theta, \epsilon)\Bigr)$ is a continuous
function of $\Theta$. $E_{gr}(\Theta)$ decreases monotonically at the points
where $E_{gr}(\Theta) <0$. Because $E_{gr}(\epsilon) = 0$ there must exist
$\Theta_0 \in (\epsilon, \Theta_{cr})$ such that $E_{gr}(\Theta) <0$ for $\Theta
\in (\Theta_0 , \Theta_{cr})$ and $E_{gr}(\Theta_0) = 0$.

Summarizing, $H(\Theta_0 , \epsilon)$ is at the 3--particle coupling constant
threshold. By Theorem~\ref{Q@th:1} $H(\Theta_0 , \epsilon)$ has a zero energy
resonance but not a zero energy bound state. If $\psi_{gr}(\Theta, \xi) \in
L^2(\mathbb{R}^6)$ is a wave function of the ground state defined on the
interval $(\Theta_0, \Theta_{cr})$ then for $\Theta \to \Theta_0 + 0$ the wave
function must totally spread (see Sec.~\ref{sec:2}). Which means that for any
$R>0$
\begin{equation}
\lim_{\Theta \to \Theta_0 + 0} \int_{|\xi| < R}  |\psi_{gr}(\Theta, \xi)|^2 d\xi
\: \to 0 .
\end{equation}
Remember also \cite{klaus2} that if the particles 1 and 2 would be bound with the energy $e_{12} < 0$ then 
the 3--particle system cannot have a square integrable ground state wave function at the energy $e_{12}$.
Though in the paper we restricted our analysis to the case of non--positive pair interactions, with additional effort one can show that 
the main results also hold without this restriction. 

In physics there is now an increased interest to the systems, which exhibit unusually large spatial extension and form the so--called \textit{halo}. 
Under halo one usually means \cite{fedorov} that the substantial part of the wave function 
is located in the classically forbidden region so that some interparticle distances exceed by far the range of the interaction. The interest in such systems 
started with the study of light atomic nuclei but the concept has now penetrated atomic and molecular physics \cite{fedorov}.A typical example 
of halo systems are weakly bound nuclei $^6$He and $^{11}$Li, where one finds a pronounced three--particle structure consisting of a tightly bound cluster 
($^4$He and $^9$Li respectively) and two neutrons. These nuclei treated as a three--particle system form 
the so--called Borromean structure (the term originates from the Italian heraldic), which means that if one of the particles are removed, the remaining two fall apart. 
The two neutrons orbiting around the core form a halo and the effective size of these systems is by far larger than that of normal stable nuclei having 
nearly the same mass. 

Numerical calculations \cite{vaagen} showed that a deeply lying resonance in the two--neutron interaction is important in reproducing the neutron halo. Even such a 
``naive'' model, where two neutron are treated as a bound particle called \textit{dineutron} \cite{vaagen} is still effectively being used today. 
The extensive three--body calculations often approximate the neutron--neutron interaction by a simple Gaussian, where the parameters are tuned so as to accommodate 
a low lying resonance. In the physics literature (see, for example \cite{vaagen,fedorov}) one often uses the asymptotic of the bound state wave function due to Merkuriev 
\cite{merkuriev} $\psi \simeq \rho^{-5/2} e^{-k\rho}$, where $\rho$ is the hyperradius and $k$ is proportional to square root of the binding energy. 
Here one should be warned against relying on the validity of this asymptotic behaviour near the threshold. 
The results obtained here show that this can be misleading. Indeed, the normalized sequence of functions 
$c_n  \rho^{-5/2} e^{-k_n \rho} $, where $c_n := \|\rho^{-5/2} e^{-k_n \rho}\|^{-1}$
totally spreads in the limit of vanishing binding $k_n \to 0$. As we know now the wave function would not totally spread unless 
one pair of particles would have a zero energy resonance. It is worth mentioning that exactly at the zero energy threshold the wave function does not have an exponential fall off. 
From the Green's function bound \cite{klaus2} it follows that 
$\psi_{gr} \geq \rho^{-4}$, where $\psi_{gr}$ is the normalized ground state at zero energy threshold, which can be chosen positive. The results presented here 
contribute to setting the general theory of halos on a rigorous footing. 

Another example of spatially extended Borromean structures are the so--called Efimov states. 
The Efimov states predicted by V. Efimov \cite{vefimov}
attracted considerable interest due to their bizarre and counter-intuitive
properties. These states start to appear when at least two of the binary subsystems 
either have very large scattering lengths or bound states at nearly zero energy. In the limiting case when 
 at least two of the binary subsystems have zero energy resonances the number of such states is infinite. In that case the binding energy of the $n$-th state 
decreases exponentially with $n$, and bound states attain enormous spatial extension. (The infinite sequence of Efimov states $\psi_n$ totally spreads because 
$\psi_n \wto 0$ due to orthogonality of the states and $\sup_n \|H_0 \psi_n\|< \infty$, see Lemmas~1, 3). 
The existence of this effect 
was demonstrated rigorously by Yafaev in \cite{yafaev}, see also \cite{sobol}. Remarkably, the three--particle system has an infinite number of bound states in spite 
of the fact that all its subsystems are unbound. 
These states evaded any experimental evidence for 35 years 
since their prediction until Kraemer \textit{et al.} \cite{efimov} reported on their discovery in an ultracold gas of cesium atoms. 
The present paper predicts extended halo--like states for three atoms near zero energy threshold, if one pair of atoms has a large scattering length (that is it is close 
to the zero energy resonance). Therefore, we advise experimentalists and theoreticians to look for such states in ultracold gas mixtures prepared through the appropriate 
Feshbach resonance tuning \cite{chin}. In analogy with Efimov states they can be, probably, indirectly detected through the 3--body recombination loss.

\ack 

The author would like to thank Prof. Walter Greiner for the warm hospitality at FIAS.

\end{document}